\documentclass[runningheads]{llncs}
\usepackage{times}
\usepackage{latexsym,amssymb,url}

\usepackage{tikz}
\usetikzlibrary{arrows}

\newtheorem{CLAIM}{Claim}
\renewenvironment{claim}{\begin{CLAIM} \hspace{-.85em} {\bf .} \rm}%
  {\end{CLAIM}}

\newcommand{\remove}[1]{}

\newcommand{\union}{\cup}

\newcommand{\set}[1]{\{ #1 \}}

\newcommand{\nbit}{\bit^n}

\newcommand{\bit}{\set{0,1}}

\newcommand{\MYPROD}{\prod^{\ell}}
\newcommand{\MYPRODG}{\prod^{\ell}_{G}}
\newcommand{\MAP}{{\rm MAP}}
\newcommand{\DISJ}{{\rm DISJ}}
\newcommand{\NOF}{{\rm NOF}}
\newcommand{\MYOPIC}{{\rm MYOPIC}}
\newcommand{\VIEW}{{\rm VIEW}}
\newcommand{\SYM}{{\rm SYM}}
\newcommand{\LEN}{{\rm LEN}}

\newcommand{\eqref}[1]{Equation~(\ref{#1})}

\begin{document}
\title{On the Power of Multiplexing \\ in Number-on-the-Forehead Protocols}

\author{Eldar Aharoni \and {Eyal Kushilevitz}}

\institute{{Department of Computer Science, Technion, Haifa, Israel \\
\email{\{eldara|eyalk\}@cs.technion.ac.il}}
}

\maketitle

\begin{abstract}
We study the direct-sum problem for $k$-party ``Number On the Forehead'' (NOF) deterministic communication complexity.
We prove several positive results, showing that the complexity of computing a function $f$ in this model, on $\ell$ instances, may be significantly cheaper than $\ell$ times the complexity of computing $f$ on a single instance. Quite surprisingly, we show that this is the case for ``most'' (boolean, $k$-argument) functions.
We then formalize two-types of sufficient conditions on a NOF protocol $Q$, for a single instance, each of which guarantees some communication complexity savings when appropriately extending $Q$ to work on $\ell$ instances.
One such condition refers to what each party needs to know about inputs of the other parties, and the other condition, additionally, refers to the communication pattern that the single-instance protocol $Q$ uses.
In both cases, the tool that we use is ``multiplexing'': we combine messages sent in parallel executions of protocols for a single instance, into a single message for the multi-instance (direct-sum) case, by xoring them with each other.
\end{abstract}

\section{Introduction}

Communication Complexity,
introduced by Yao‎~\cite{Yao79}, studies the amount of communication parties need to exchange in order to carry out a distributed computation. In its most basic setting (so-called two-party communication complexity) two parties, Alice and Bob, each having $n$-bit input ($x$ and $y$ respectively) want to compute the value $f(x, y)$, for some fixed $f$. To do so, they use a predefined communication protocol $Q$. The cost of $Q$ is the number of bits exchanged in $Q$, for the worst choice of $(x, y).$ The communication complexity of a function $f$, denoted $D(f)$, is the minimal cost of an appropriate $Q$. The field deals with protocols for interesting functions and mostly with lower bounds for them, and has a wide range of applications (see~\cite{KN} for an overview).

In the multi-party case, there are several common models.
The ``number on the forehead'' ($\NOF$) model of communication complexity, which is studied also in the current paper, was introduced by Chandra, Furst and Lipton~\cite{CFL83}. In this model, $k$ parties wish to jointly evaluate a function $f : (\nbit)^{k} \rightarrow \bit$, where each input of $f$, i.e. $x_i \in \nbit$ for $i\in[k]$, is given to all parties {\em except} $P_i$ (pictorially, it is viewed as if $x_i$ appears on the forehead of $P_i$). In order to compute $f$, the parties communicate by writing bits on a shared board. The communication complexity of a NOF protocol, is the number of bits written on the board, for the worst input. The $\NOF$ model is not as well understood (as, for example, the two-party case), mainly because of the complexity of dealing with the shared inputs, as every piece of the input is seen by all the parties but one. An alternative multi-party model is the so-called ``Number in Hand'' model, where each $x_i$ is given to party $P_i$ (see, e.g.,~\cite{DKW,PVZ}).
Even though the ``Number in Hand'' model
may seem more natural, results for the $\NOF$ model have many applications in areas like branching programs complexity, boolean circuit complexity, proof complexity and pseudorandom generators. In all these applications, the overlap of information in the NOF setting, models well the problem in question. Note that for $k = 2$, both models collide with the standard two-party model.
We next reference several papers, which studies $\NOF$ communication complexity. Separation of classes of complexity, within the $\NOF$ communication complexity, is examined, for example, in~\cite{BDPW10} and~\cite{DPV08}. The $\NOF$ communication complexity of the disjointness function (defined in the following subsection) is studied in~\cite{BPSW06}~\cite{CA08}~\cite{LS08} and~\cite{S12}.

\medskip

Direct-sum questions usually refer to the task of simultaneously solving multiple instances of the same problem, and ask whether this task can be done more efficiently than simply solving each instance separately. For example, the task of computing a function $f$ on $\ell$ instances, is a direct-sum problem denoted $f^\ell$. Direct-sum problems for two-party communication complexity have been first studied in‎~\cite{KRW95}, where the authors showed that if a certain direct-sum result holds for deterministic communication complexity of relations, then {\rm{P}} $\neq$ {\rm{NC}}$^1$.
In~\cite{FKNN95}, the authors proved that $D(f^\ell) \geq \ell \cdot (\sqrt{D(f)/2} - 2 \log n - O(1))$ and presented a partial function $f$ such that $D(f) = \Theta (\log {n})$, whereas $D(f^\ell) = O(\ell + \log n \cdot \log \ell)$. They also proved that for any (even partial) function $f$, for one-round protocols, $D^1(f^\ell) \geq \ell \cdot (D^1(f) - O(\log {n}))$. In~\cite{KKN92}, the same is proved for the two-round deterministic case. Namely, in these cases no significant savings is possible. For general two-party deterministic protocols, this fundamental question (stated as a conjecture in~\cite{ABG+01}) remains open:
\begin{question}
Is it true that, for all functions $f :\nbit \times \nbit \rightarrow \bit$, only a minor saving is possible; formally, $D(f^{\ell}) \geq \ell \cdot (D(f) - O(1))$?
\end{question}

While this is open for $k=2$, we prove that the analogue question for the NOF model has negative answer, for every $k > 2$. Namely, savings is ``usually'' possible.

\subsection{Our Results}

Our first result shows that savings is possible for the direct-sum problem of ``most'' boolean functions, in the deterministic
$k$-party $\NOF$ model.
Namely, assume for convenience that $(k - 1)$ divides $\ell$, then $D^{\NOF}(f^{\ell})\leq \ell\cdot \frac{n}{k - 1} + \ell$, for all boolean functions $f$. Since ``most'' boolean functions have high communication complexity, specifically $D^{\NOF}(f)\ge n-\log(k+2)$~\cite{G93}, then
for most functions our upper bound  is significantly cheaper than $\ell$ times the complexity of computing $f$ on a single instance.

The above upper bound turns out not too difficult to prove. The main idea is that of appropriately {\em multiplexing} messages sent in various executions of the underlying protocol for $f$.

We proceed to further understanding the power of multiplexing. We formulate conditions under which multiplexing is possible.
For this, we define a variant of the $\NOF$ model, in which parties can view only a subset of other parties' inputs. We give sufficient conditions on $\ell$ such protocols that guarantee that, instead of simply running the $\ell$ protocol simultaneously, we can appropriately combine them into one (standard) $\NOF$ protocol that solves all $\ell$ instances, and this leads to some savings in communication. Next, we present sufficient conditions on a single protocol for $f$, which allow us extending it to a $\NOF$ protocol that solves $\ell$ instances of the original problem, with appropriate savings in communication complexity. The conditions used are technical, but we show, for instance, that for the important case of symmetric functions, they can be
significantly simplified. We give examples of appropriate $\NOF$ protocols that are constructed by applying our technique.

We also use the so-called $\MYOPIC$ variant of the NOF model~\cite{G06,C07}, which limits the way the inputs are shared between the parties and, additionally, limits the communication pattern between the parties to be one-way. We give conditions on $\ell$ $\MYOPIC$ protocols, which allow extending them to a $\NOF$ protocol solving $\ell$ instances, with appropriate savings in communication complexity.

\subsection{Related Work}

The direct-sum problem was also studied with respect to other communication complexity measures.
For the non-deterministic case, it was shown~\cite{FKNN95}
that $N(f^\ell) \geq \ell \cdot (N(f) - \log n - O(1))$.  An alternative proof was given in~\cite{KKN92}.
For the randomized case, it is shown~\cite{FKNN95} that some savings is possible; specifically for $EQ$, the equality function, $R(EQ^{\ell}) = O(\ell + \log n)$, whereas $R(EQ) = \Theta(\log n)$. Using information-complexity arguments, Barak et al.~\cite{BBCR09} showed that (ignoring logarithmic factors) computing a function $f$ on $n$ instances, using a randomized protocol, requires
$\Omega(\sqrt n)$ times the randomized communication complexity of $f$.

Not much is known regarding the direct-sum problem $f^\ell$, in the
$\NOF$ model.
In ~\cite{BPSW06}, the authors state that it is not clear how to prove a direct sum theorem for $\NOF$. Certain direct-sum like operators, for the case of $n = 1$, were considered in~\cite{ACFN12}. The authors proved that for any symmetric boolean function $\SYM : {\bit}^\ell \rightarrow \bit$ (that depends only on the number of ``1'' inputs), and for any boolean function $f_k : {\bit}^k \rightarrow \bit$, the communication complexity $D(\SYM(f_{k}^\ell)) = O(\frac{\ell}{2^k} \cdot \log \ell + k \log \ell).$ The
proof works for $k > 1 + \log \ell$.
It is also shown that $D^{||}(\SYM(f_{k}^{\ell})) =O((\log \ell)^3)$. Here, $D^{||}$ refers to the communication complexity in the simultaneous $\NOF$ model, where each party, without interaction, sends a single message to an external referee, who computes the output based only on the messages it received.

For randomized protocols, a strong direct product theorem for a problem, states that when bounding the communication complexity to be less than, say, $\ell$ times the communication complexity required for solving the problem on one instance, with constant error probability, then the probability of solving $\ell$ instances of the problem with constant error, is exponentially small in $\ell$. Strong direct-product results in the two-party model, were shown for the one-way public-coin setting~\cite{J11} and bounded round public-coin‎ setting~\cite{JPY12}.

One of the most studied functions in communication complexity is the disjointness function, $\DISJ$. Its $k$-argument variant $\DISJ_{k}$ is $1$ iff, there is no $i \in [n]$ that belongs to all the $k$ sets. That is, the inputs, viewed as sets, are disjoint.
In one of the few results concerning direct-product in the $\NOF$ model, Sherstov~\cite{S12} showed that $\ell \cdot \Omega(\frac{n}{4^k})^{\frac{1}{4}}$ bits are required to solve $(1 - \epsilon) \cdot \ell$ out of $\ell$ instances of the function $\DISJ_{k}$, for small enough $\epsilon > 0$. Clearly, solving all $\ell$ instances correctly, will take at least that much communication.

In~\cite{ABG+01,BDKW10}, the communication complexity of other direct-sum like ``operators'' on multiple instances was studied. For example, the problem of choosing one instance from $\ell$ given instances and solving that instance. Alternatively, eliminating one $\ell$-tuple of outputs which is \emph{not} the correct one. In most of these cases, demonstrating savings turns out to be difficult.

\section{Direct-Sum in the NOF Model -- Basic Examples}

As a first example, we show that savings is possible for the direct-sum problem of most boolean functions in the deterministic communication complexity $\NOF$ model. Specifically, for all functions whose NOF complexity is ``high''. From now on, we denote by $x_{i, j}$ the input on $P_j$'s forehead in the $i^{th}$ instance.

\begin{lemma}
\label{lem-nof-ds}
Let $f_{k}: (\nbit)^k \rightarrow \bit$ be an arbitrary $k$-argument boolean function. Then,
$D^{\NOF}(f_{k}^{k-1})\leq n + k-1$.
\end{lemma}

\begin{proof}
The idea is to use, what we term, ``multiplexing''. We present a protocol $Q$ for the problem: $P_{k}$ computes $M=\oplus_{i=1}^{k-1} (x_{i,i})$. That is, it xores the input of $P_1$ for instance $1$ with the the input of $P_2$ for instance $2$, and so on for each of the $k-1$ instances. Party $P_{k}$ writes this value on the common board. Now, for $1 \leq i <k$, party $P_{i}$ knows all the values used by $P_k$ to compute $M$, except for $x_{i,i}$ so, by xoring $M$ with those $k-2$ values, $P_i$ can determine $x_{i,i}$. This, in turn, means that now $P_i$ knows all the inputs for instance $i$, and so it locally computes the output of $f$ on the $i$-th instance and writes this value on the board. As $f$'s value on each instance, has been written by the corresponding $P_{i}$ ($1 \leq i <k$), then all parties now know the correct output for each of the $k-1$ instances.
\qed
\end{proof}

If the number of instances, $\ell$, is different than $k-1$, we simply partition the instances into ``blocks'' of size $k-1$. Hence,

\begin{corollary}
\label{corol-nof-ds}
Let $f_{k}: (\nbit)^k \rightarrow \bit$ be an arbitrary $k$-party boolean function and $(k - 1)$ divides $\ell$, then $D^{\NOF}(f_{k}^{\ell})\leq \ell \cdot \frac{n}{k - 1} + \ell$.
\end{corollary}

\medskip
In~\cite{G93}, it is proved that most $k$-argument functions $f: (\nbit)^{k} \rightarrow \bit$, have communication complexity $D^{\NOF}(f) \geq n - \log (k + 2)$, in the $\NOF$ model (see also~\cite[Exercise 6.5]{KN}). Using Corollary~\ref{corol-nof-ds} with such a function, gives a provable savings for the direct-sum problem in the deterministic $\NOF$ model.
Since most functions are also hard with respect to other complexity measures in the NOF model (e.g., randomized complexity, as showen in ~\cite{JF06}), Corollary~\ref{corol-nof-ds} implies a similar result for those as well.

\medskip
The above example deals with savings for ``hard'' functions. Next, we show how multiplexing can (slightly) help in cases where non-trivial protocols exist for a single instance.
Specifically, we consider the boolean function ``equality'', which is defined as $EQ_k(x_1, \ldots, x_k) = 1$ iff $x_1 = \ldots = x_k$.

\begin{proposition}
\label{prop_eq}
$D^{\NOF}(EQ_k) = 2$.
\end{proposition}

\begin{proof}
We present a protocol for the problem. Party $P_{k-1}$ writes on the board $1$ if $x_{k-2} = x_k$, and $0$ otherwise. If $P_{k}$ sees $1$ on the board and if all the inputs it sees are equal, then $P_{k}$ writes as output ``$1$'', otherwise it writes ``$0$''. Obviously, the protocol is correct, so $D^{\NOF}(EQ_k) \leq 2$.
For the lower bound, the first party $P_i$ to write a bit on the board (the index of this party is fixed by the protocol), must sometimes also read something from the board to know the value of $EQ_k$ (in case that all other inputs $x_j$ are equal), so that requires at least two bits of communication.
\qed
\end{proof}

\begin{proposition}
\label{prop_eq_multiple}
For odd $k \geq 3$, we have $D^{\NOF}(EQ_{k}^{\frac{k - 1}{2}})\leq 1 + \frac{k - 1}{2}$.
\end{proposition}

\begin{proof}
We give a protocol for $EQ_{k}^{\frac{k - 1}{2}}$: For odd $j$'s, define
$$b_{j} = \left\{
\begin{array}{ll}
		1 & \mbox{if } x_{j, j} = x_{j, j + 1} \\
		0 & \mbox{otherwise }
\end{array}
\right. $$

\begin{itemize}
\item $P_{k}$ computes the xor of all the bits $b_{j}$, for $j=\{1, 3, \ldots, k - 2\}$, which we denote as $M$, and writes in on the board.
\item For $j \in [1,\frac{k - 1}{2}]$, party $P_{2 \cdot j}$ knows all the values used to compute $M$ except $b_{2 \cdot j - 1}$, so it can now compute it. Then, if all the inputs in instance $2 \cdot j$ that $P_{2\cdot j}$ sees are the same and $b_{2 \cdot j - 1} = 1$, then it writes $1$ on the board as the output for this instance; otherwise, it writes $0$.
\end{itemize}
The communication is $1 + \frac{k - 1}{2}$ bits (i.e., better than $\frac{k - 1}{2} \cdot D^{\NOF}(EQ_k) = k - 1$.)
\qed
\end{proof}

\medskip

\section{Multiplexing NOF protocols via a Restricted NOF Model}

The above examples already show the power of multiplexing.
In this section, we will describe how to multiplex the executions of a single-instance NOF protocol $Q$ on $\ell$ instances, assuming that $Q$ meets certain conditions that would allow it to run in a ``restricted NOF'' model. Before formally defining the restricted NOF model and how it is used, we
describe, informally, the basic ideas.

As a starting point, consider a $3$-party $\NOF$ protocol $Q$, for a boolean function $f \in (\nbit)^3 \rightarrow \bit$. Assume that the first message of $Q$ is written on the board by $P_1$, and that it is of the form $m_2(x_2) || m_3(x_3)$, where $||$ stands for concatenation. Assume also that $m_2, m_3$ are of the same length (for all inputs).
That is, $P_1$'s message has two equal-length parts, one that depends on input $x_2$ and one that depends on input $x_3$. Obviously, $P_2$ can compute by itself the message $m_3(x_3)$, as it sees $x_3$, and, similarly, $P_3$ can compute $m_2(x_2)$. Therefore, in this specific case, instead of writing both these messages, $P_1$ could have simply written on the board $m_2(x_2) \oplus m_3(x_3)$ (which is of half the size), and both $P_2$ and $P_3$ could have determined the original message they were supposed to get.
This can be extended beyond the first message of the protocol: namely, when $m_2$ depends on $x_2$ and on the communication so far, and $m_3$ depends on $x_3$ and on the communication so far. Since both $P_2$ and $P_3$ see the board themselves, they can still compute the values $m_3,m_2$, respectively.

Next, we wish to extend the above simple example to more parties; specifically, consider the case $k=4$.
Suppose that $P_1$ would like to write on the board the expression $m_2(x_2, x_4) || m_3(x_3, x_4)$. In this case, $P_2$ and $P_3$ will still be able to determine $m_2,m_3$ even if $P_1$ will write on the board the value $m_2\oplus m_3$. However, $P_4$ (who does not know $x_4$) cannot compute $m_2, m_3$ from $m$. In some cases, however, $P_4$ itself does not need the values $m_2.m_3$ to continue its computation.
To model what information is needed by each party, it is convenient to view the $\NOF$ protocol $Q$ as a protocol $Q'$ in a model where communication is over point-to-point channels (where a party that does not need a message, simply does not receive it).
Say that in $Q'$ party $P_1$ sends a message $m_2(x_2, x_4)$ to $P_2$ and a message $m_3(x_3, x_4)$ to $P_3$. We then can get a new $\NOF$ protocol $Q''$ by multiplexing this message, and so we save the length of a single message. In $Q''$, party $P_1$ will write on the board  $M = m_2(x_2, x_4) \oplus m_3(x_3, x_4)$. Now, $P_2$ and $P_3$ will be able to determine their original messages, as explained before, and since we assumed that in $Q'$, party $P_4$ did not receive a message at this stage, its future messages in $Q'$ do not depend on messages $m_2, m_3$ and so $P_4$ can simply ignore the value $M$ on the board in $Q''$.

Now that we have moved to the point-to-point model, we need to reconsider what if the messages we wanted to multiplex was not the first message to be sent. For instance, suppose we take a protocol which begins with $P_2$ and $P_3$ sending a message, each, to $P_1$. Then, $P_1$ sends a message $m_2(x_2, H)$ to $P_2$ and a message $m_3(x_3, H)$ to $P_3$, where $H$ stands for the history of messages that $P_1$ received. Now, if $P_1$ will multiplex these messages, then $P_2$ may not be able to compute $m_3(x_3, H)$ as it may have only partial information on $H$ (since we are dealing now with a point-to-point model, $P_2$ did not see the message sent from $P_3$ to $P_1$). However, when having, say, two instances for $f$, we can sometimes still do something. Below, we formalize constrains for when two (equal length) messages $m_1$ in protocol $Q_1$ and $m_2$ in $Q_2$ can be multiplexed. Obviously, the sender of these two messages should be the same, say $P_i$, and the recipients need to be different, say $P_{j_1}$ and $P_{j_2}$. Denote these messages by $m_{1, i \rightarrow j_1}$ and $m_{2, i \rightarrow j_2}$. In order for $P_{j_1}$ to be able to compute the message $m_{2, i \rightarrow j_2}$ (in order to ``de-multiplex'' the message on the board), it is enough that~(1) this message does not depend on $x_{j_1}$, and~(2) that $P_{j_1}$ knows the history of messages received by $P_i$ so far (similar requirements needed for $P_{j_2}$). To fulfil the second requirement, we require that none of the messages sent to $P_i$ in $Q_1$ and $Q_2$ will be multiplexed (so in the resulting $\NOF$ protocol, those messages will simply appear on the board and $P_{j_1}$ will know them).
This discussion can be simply extended to deal with $\ell$ protocols for which we want to multiplex messages from $P_i$ to $P_{j_{1}}, \ldots, P_{j_{\ell}}$, respectively.

To force protocols to satisfy the above requirements, we define below a new model, where instead of having requirements on messages sent from parties, we simply limit the parties' view of other parties' inputs.
This is defined using a directed graph $G$, that describes what inputs are seen by each party. Still, communication is allowed between any pair of parties over a private channel.

\begin{definition}
Let $G = (V, E)$ be a simple directed graph over $V = \{1, \ldots , k \}$. A restricted $\NOF_{G}$ multi-party model with $k$ parties has the restriction that, for all $i,j\in[k]$, party $P_i$ sees the input of party $P_j$ iff $(i,j) \in E$. The graph $G$ will be referred to as the restriction graph. Communication in the $\NOF_{G}$ model is over a complete network of point-to-point channels. Namely, in a $\NOF_{G}$ protocol $Q$, in each round $t$, each party $P_i$ can send a message to each party $P_j$ depending on messages $P_i$ received in the previous $t - 1$ rounds and depending on the inputs $P_i$ sees. (We emphasize that we allow $P_i$ sending a message to $P_j$ even though $(i, j) \notin E$). We say that protocol $Q$, in the $\NOF_{G}$ model, computes $f \in ({\nbit})^k \rightarrow \bit$, if some party $P_i$, whose identity is known to all parties, outputs $f(\vec{x})$.\footnote{This is in contrast to the $\NOF$ model, where all the parties output $f(\vec{x})$.}
\end{definition}

We will use the following notations. For a node $i \in V$ denote the neighbors of $i$, by $N(i) = \{j | (i, j) \in E \}$.
For an input $\vec{x}$, denote $\VIEW_{i, G}(\vec{x})$ as the inputs that $P_i$ sees, that is, the sequence $\{x_j | j \in N(i)\}$. Denote $m_{j \rightarrow i,t,\vec{x}}$ as the message sent from $P_j$ to $P_i$ in round $t$ on input $\vec{x}$. Let $M_{i, t}(\vec{x}) = \{(t', j, m_{j \rightarrow i, t',\vec{x}}) | t' \leq t, j \in [k]\}$, be the history of messages that $P_i$ received in rounds $1$ to $t$ on input $\vec{x}$. The message sent from $P_i$ to $P_j$ in round $t+1$ in protocol $Q$ may depend on the messages $P_i$ received in previous rounds and on the inputs seen by $P_i$ (as set by $G$). That is, it can be denoted as $m^{Q}_{i \rightarrow j, t + 1,\vec{x}}(M_{i, t}(\vec{x}), \VIEW_{i, G}(\vec{x}))$. When the input $\vec{x}$ is clear from the context, we will usually omit it from the notations.

\begin{definition}
An oblivious $k$-party $\NOF_{G}$ protocol $Q$ is a protocol in which, for every input vector $\vec{x}$, the length of the messages exchanged is fixed. That is, there exist a mapping $\LEN^{Q}:\mathbb{N} \times [k] \times [k] \rightarrow \mathbb{N}$, referred to as the {\em communication pattern of $Q$}, that maps $(t, i, j)$ to the length of the message sent from $P_i$ to $P_j$ in round $t$, on every $\vec{x}$.
For an oblivious protocol $Q$, we denote the number of bits sent from $P_a$ to $P_b$ in protocol $Q$ by $W_Q(a, b)$.
\end{definition}

The following definitions will capture sufficient requirements from restricted $\NOF$ protocols, so that it is possible to multiplex messages when combining the protocols together.

\begin{definition}
\label{def:good}
Let $G = (V, E)$ be a restriction graph for $k$ parties. Denote the isomorphic to $G$ graph, which is the resulted graph of permutating $G$'s nodes with a permutation $\pi$, by $G_{\pi}$. Let $\overline{N_{G}(a)}$ be the non-neighbours of $a$ in the graph $G$, that is, $V\backslash N_{G}(a)$. Let $\MYPROD = \pi_1, \ldots, \pi_\ell$ be a sequence of $\ell$ permutations over $[k]$. For $b \in [k], R \subseteq [\ell]$, denote $\MAP_{\MYPROD, b, R} = \left\{\pi_{r}(b) : r \in R \right\}$. That is, the images of $b$ under the permutations in $R$.  A triplet $(a, b, R)$, for $a \neq b \in [k]$ and $R \subseteq [\ell]$, is $\MYPRODG$-good if:
\begin{enumerate}
\item $\MAP_{\MYPROD, a, R}=\{a\}$.
\item $\MAP_{\MYPROD, b, R}$ is a subset of $\overline{N(a)}$ of size $|R|$, that does not include $a$ or $b$. \footnote{Condition~(1) already implies that $a \notin \MAP_{\MYPROD, b, R}$, as $a$ is already mapped to $a$ by all the permutations in $R$. This requirement is included in condition~(2) for convenience.}
\item For all $r \in R$ we have $(\MAP_{\MYPROD, b, R} \bigcup \{b\} \backslash \{\pi_r(b)\}) \subseteq \overline{N_{G_{\pi_r}}(a)}$.
\end{enumerate}
\end{definition}

Let us discuss the above definition. Suppose we have a $\NOF_G$ protocol $Q$ that we want to multiplex. Intuitively, each triplet $(a, b, R)$ is a candidate for multiplexing.
The messages that will be multiplexed, are the message that $P_a$ sends to $P_b$ in $Q$, and the messages $P_{\pi_r(a)}$ sends to $P_{\pi_r(b)}$, for each permutation $\pi_r$, for $r \in R$. Condition~(1) guarantees that $a$ is mapped to itself in all those permutations, so it is the only sender of the messages we want to multiplex together, as required for multiplexing.
Condition~(2) requires that $b \notin \MAP_{\MYPROD, b, R}$ and $\MAP_{\MYPROD, b, R}$ being the size of $|R|$,  which guarantees that we will have $|R| + 1$ messages to multiplex (messages which are sent from $P_a$ to $P_b$, or to other ``alternative'' recipients, as determined by $R$).
As stated before, for multiplexing to work, any multiplexed message sent to a certain recipient, should not depend on any of the other recipients' input. Conditions~(2) (i.e., that $\MAP_{\MYPROD, b, R}$ is a subset of $\overline{N(a)}$) and condition~(3) take care of that.

\begin{definition}
\label{def:multiplexing}
Let $G, k$, and $\MYPROD$ be as in Definition~\ref{def:good}. A $\MYPRODG$-multiplexing set $S$ is a collection of $\MYPRODG$-good triplets, such that for any two distinct $(a, b, R), (a', b', R') \in S$ we have $a \notin  (\{b'\} \bigcup \MAP_{\MYPROD, b', R'})$ and $(\{b\} \bigcup \MAP_{\MYPROD, b, R}) \bigcap (\{b'\} \bigcup \MAP_{\MYPROD, b', R'}) = \emptyset$.
\end{definition}

Intuitively, a $\MYPROD$-multiplexing set $S$ guarantees that if we multiplex any messages sent from $a$, we will not multiplex messages sent {\em to} $a$. Also, we will not have a recipient or an ``alternative'' recipient, for the same $a$, more than once.

The next definition formalizes the ``communication pattern'' of an oblivious protocol. Having the same ``communication pattern'' between protocols is necessary for our multiplexing stage in the following proofs.
\begin{definition}
Let $Q$ be an oblivious $k$-party $\NOF_{G}$ protocol (for some restriction graph $G = (V, E)$), which computes a $k$-argument boolean function $f$. For a communication pattern $\LEN$, we denote $\LEN_{\pi}$ as the mapping $\LEN_{\pi}(t, i, j) = \LEN(t, \pi^{-1}(i), \pi^{-1}(j))$. Protocol $Q$ is $\pi$-pattern-robust, if there is an oblivious protocol $Q'$ for computing $f$ in $\NOF_{G_{\pi}}$, such that $\LEN^{Q'} = \LEN^{Q}_{\pi}$. That is, if there is a $\NOF_{G_{\pi}}$ protocol $Q'$ for computing $f$, such that the communication pattern is the same for the two protocols, up to the permutation $\pi$.
\end{definition}

\begin{theorem}
\label{thm:second}
Let $f$ be a $k$-argument boolean function, $G = (V, E)$ a restriction graph for $k$ parties, $Q$ an oblivious $k$-party $\NOF_G$ protocol for $f$ and $\MYPROD = \pi_1, \ldots \pi_\ell$ be a sequence of permutations over $[k]$ such that $\pi_1$ is the identity permutation. Also, assume that, for $2 \leq i \leq \ell$, protocol $Q$ is $\pi_{i}$-pattern-robust, and $S$ is a $\MYPRODG$-multiplexing set. Then, $$D^{\NOF}(f^{\ell}) \leq \ell \cdot D^{\NOF_G}(Q) + \ell - \sum_{(a, b, R) \in S} |R| \cdot W_Q(a, b).$$
\end{theorem}

Intuitively, the $\NOF$ communication complexity of $f^{\ell}$ is bounded by $\ell$ times the $\NOF_{G}$ communication complexity of $f$, plus $\ell$ bits to distribute the outputs, minus the savings, which  depends on the multiplexing set. For any $(a, b, R)$ in the multiplexing set, we save $|R|$ times the communication between $a$ and $b$ in the protocol $Q$.

\medskip

\begin{proof}
Protocol $Q$ is $\pi_i$-pattern-robust, for $i = 2, \ldots, \ell$. Therefore, there are protocols $Q_{\pi_i}$ such that $Q_{\pi_i}$ solves $f$ in $\NOF_{G_{\pi_i}}$ and $\LEN^{Q} = \LEN^{Q_{\pi_i}}$. We take those $\ell$ protocols and combine them into a single $\NOF$ protocol $Q'$ for computing $f^\ell$, where the $i^{th}$ protocol is responsible for computing $f$ on the $i^{th}$ instance. Whichever messages are sent in $Q_{\pi_1}, \ldots, Q_{\pi_\ell}$ over private channels are now communicated on the board, so clearly the computation of each message remains the same (as the messages previously received by a party, are now known to all). According to this definition of $Q'$, for each $(a_i, b_i, R_i) \in S$, the bits initially sent in round $t$ of protocol $Q$ from $P_{a_i}$ to $P_{b_i}$, are now translated into $|R_i| + 1$ messages. That is, one message from $P_{a_i}$ to $P_{b_i}$ in $Q$ itself, and one message from $P_{a_i}$ to $P_{\pi_{r}(b_i)}$ in $Q_{\pi_r}$, for $r \in R_i$. Instead of sending all these messages in $Q'$, we perform the following multiplexing  step for every round $t$: Party $P_{a_i}$ writes on the board
$$m^{Q'}_{(a_i, b_i, R_i),t} \triangleq m^{Q}_{a_i \rightarrow b_i,t} \oplus \bigoplus_{r \in R_i}(m^{Q_{\pi_r}}_{a_i \rightarrow \pi_{r}(b_i),t}).$$
We need to prove that (i) for each $r \in R_i$, party $P_{\pi_{r}(b_i)}$ can infer the original message sent to it in $Q_{\pi_r}$, and that (ii) $P_{b_i}$ itself can retrieve the message originally sent to it in $Q$.
We start with (i). To prove that $P_{\pi_{r}(b_i)}$ can compute the message $m^{Q'}_{a_i \rightarrow \pi_{r}(b_i),t}$, we need to show that:
\begin{enumerate}
\item All the xored messages (beside $m^{Q_{\pi_r}}_{a_i \rightarrow \pi_{r}(b_i),t}$) do not depend on the input of $P_{\pi_{r}(b_i)}$ (note that since $Q'$ is in the regular $\NOF$ model, $P_{\pi_{r}(b_i)}$ knows all the inputs but its own).
\item $P_{\pi_{r}(b_i)}$ knows the history of messages received by $P_{a_i}$ in all those protocols (as the xored messages may also depend on them).
\end{enumerate}
For (1), let $r \neq r' \in R$.  We show that $P_{\pi_{r}(b_i)}$ can compute $m^{Q_{\pi_{r'}}}_{a_i \rightarrow \pi_{r'}(b_i),t}$. This follows from the fact that $\pi_r(b_i)$ belongs to $\MAP_{\MYPROD, b_i, R_i} \bigcup \{b_i\} \backslash \{\pi_r'(b_i)\}$ (by definition of $\MAP_{\MYPROD, b_i, R_i}$), and $(\MAP_{\MYPROD, b_i, R_i} \bigcup \{b_i\} \backslash \{\pi_r'(b_i)\}) \subseteq \overline{N_{G_{\pi_r'}}(a_i)}$ by the $\MYPRODG$-goodness of triplet $(a_i, b_i, R_i)$) (specifically, Definition~\ref{def:good}, condition~(4)).
We show that $P_{\pi_{r}(b_i)}$ can compute $m^{Q}_{a_i \rightarrow b_i,t}$. This follows from the fact that $\MAP_{\MYPROD, b_i, R_i} \subseteq \overline{N(a_i)}$ (by Definition~\ref{def:good}, condition~(2)).
For (2), this is true, as messages which are sent to $P_{a_i}$ in the original protocols, are not multiplexed in the combined protocol (by Definition~\ref{def:good},
for $u \neq i$, we have that $a_i \notin (\{b_u\} \bigcup \MAP_{\MYPROD, b_u, R_u})$.

The proof for~(ii) is similar, but somewhat simpler. The argument for~(2) is just the same. The argument for~(1), is as follows: Let $r \neq r' \in R$.  We show that party $P_{b_i}$ can compute $m^{Q}_{a_i \rightarrow \pi_{r'}(b_i),t}$. This is true because $b_i$ belongs to $\MAP_{\MYPROD, b_i, R_i} \bigcup \{b_i\} \backslash \{\pi_{r'}(b_i)\}$ (as $b_i \neq \pi_{r'}(b_i)$, by Definition~\ref{def:good},  condition~(2)).

An additional step in $Q'$ is the reporting of the output for each of the instances, to all parties. As the protocols $Q$ and $Q_{\pi_i}$, for $i = 2, \ldots, \ell$, all ended up with one party knowing the output, and since $Q'$ is in the $\NOF$ model, then for each $i = 1, \ldots, \ell$, the party who knows the output on instance $i$ now writes the output bit on the board. It follows that $D^{\NOF}(f^{\ell}) \leq \ell \cdot D^{\NOF_G}(Q) + \ell - \sum_{(a, b, R) \in S} |R| \cdot W_Q(a,b)$.
\qed
\end{proof}

\begin{example}
We now show that Theorem~\ref{thm:second} implies Lemma~\ref{lem-nof-ds}. Let $f$ be a $k$-argument function. Let $\ell = k - 1$. Let $G = ([k], \left\{(k, 1)\right\} \bigcup \{(1, i) : 2 \leq i \leq k\})$ be a restriction graph. I.e, the only edges in $G$ are between $k$ to $1$ and between $1$ to all the other nodes. We describe a protocol $Q$ for $f$ in $\NOF_G$: $P_{k}$ sends $x_1$ to $P_1$ and $P_1$ computes $f(\vec{x})$, as it now knows all $k$ inputs. We take permutations $\pi_2 = (2, 3, \ldots, k - 1, 1, k), \pi_3 = (3, 4, \ldots, k - 1, 1, 2, k), \ldots , \pi_\ell = (k - 1, 1, \dots, k - 2, k)$.
Let $Q^i$ be the following protocol for computing $f$ in $\NOF_{G_{\pi_i}}$: party $P_k$ sends $P_{i}$ its input, and $P_{i}$ (which now knows all the inputs) computes the value $f(\vec{x})$ and outputs it. As $\LEN^{Q^i} = \LEN^{Q}_{\pi_i}$, it follows that $Q$ is $\pi_i$-pattern-robust.
Let $S$ contain one triplet $(k, 1, \{2, \ldots, k - 1 \})$. We argue that the triplet is $\MYPRODG$-good. The value
$k$ is mapped to itself in all the permutations. The value $1$ is mapped to $k - 2$ different values under the $k - 2$ permutations $\pi_2, \ldots, \pi_l$, none of which is $1$ and all of which are not neighbours of $k$ in $G$.
We have $\MAP_{\MYPROD, 1, R} = \{2, \ldots, k - 1\}$, and for all $r \in \{2, \ldots, k - 1 \}$, we have $N_{G_{\pi_r}}(k) = \{r\}$. The condition that for all $r \in \{2, \ldots, k - 1 \}$, $(\MAP_{\MYPROD, 1, R} \bigcup \{1\} \backslash \{\pi_r(1)\}) \subseteq \overline{N_{G_{\pi_r}}(k)}$, translates to requiring that for for all $r \in \{2, \ldots, k - 1 \}: [k - 1] \backslash \{r\} \subseteq  [k - 1] \backslash \{r\}$, which is true. Then, by Theorem$~\ref{thm:second}$, it follows that $D^{\NOF}(f_{k}^{\ell}) \leq \ell \cdot n + \ell - (k - 2) \cdot n = n + \ell$.
\end{example}

Employing Theorem~\ref{thm:second}, requires having $\ell$ protocols for $f$ in $\ell$ restricted $\NOF$ models. Sometimes finding such protocols may be difficult. We next turn to a specific family of functions, specifically the symmetric functions, and see how we can have a simpler theorem that can be applied to functions in this family.

\begin{definition}
A $\pi$-symmetric $k$-argument function $f \in (\nbit)^{k} \rightarrow \left\{0, 1\right\}$ is a function whose value remains the same when applying the permutation $\pi$ to its inputs. That is, for any $\vec{x} \in (\nbit)^{k}$, it holds that $f(x_1, \ldots x_k) = f(x_{\pi(1)}, \ldots, , x_{\pi(k)})$.
\end{definition}

\begin{definition}
A symmetric $k$-argument function $f \in (\nbit)^{k} \rightarrow \left\{0, 1\right\}$ is a function who that is $\pi$-symmetric, for all permutations $\pi$ over $[k]$. That is, a function which does not depend on the order of its inputs.
\end{definition}

Let $f$ be a $k$-argument function and $Q$ a $\NOF$ protocol for $f$. Imagine each party $P_i$ has its name changed to $P'_{\pi(i)}$. Now, suppose that given input $\vec{x}$, the parties simply execute $Q$, only using their new identities, instead of the old ones. For input vector $\vec{x}$, protocol $Q'$ will correctly compute $f(\pi^{-1}(\vec{x}))$.
Now, view the original $Q$, as a $\NOF_G$ protocol for a some restriction graph $G$. The same manipulation can be applied to protocol $Q$, to get a $\NOF_{G_{\pi}}$ protocol $Q'$. This is true, as $\VIEW_{i, G_{\pi}}(\vec{x}) = \VIEW_{\pi^{-1}(i), G}(\pi^{-1}(\vec{x}))$. We formalize this below.

\begin{definition}
\label{def:per_prot}
Let $f$ be a $k$-argument function, $Q$ be a $k$-party $\NOF_G$ protocol for a given graph $G = (V, E)$ and $\pi$ a permutation over $[k]$. Denote $\pi(Q)$ as a protocol based on $Q$ for a restricted $\NOF_{G_\pi}$, such that for all $i, j, t$ the message $m^{\pi(Q)}_{i \rightarrow j, t}$ for the input $\vec{x}$ is the corresponding $m^{Q}_{\pi^{-1}(i) \rightarrow \pi^{-1}(j), t}$ for the input $\pi^{-1}(\vec{x})$.
\end{definition}

To show that this protocol is well defined, we argue that the function $m^{\pi(Q)}_{t + 1, i \rightarrow j}$ is computable by $P_i$ in $\NOF_G$. For that, we need to show that $P_i$ at round $t+1$ knows $M_{\pi^{-1}(i), t}$ and $\VIEW_{\pi^{-1}(i), G}$. This can be shown by a simple induction, where the base step ($t=1$) is correct as $\VIEW_{i, G_{\pi}}(\vec{x}) = \VIEW_{\pi^{-1}(i), G}(\pi^{-1}(\vec{x}))$ and $M_{\pi^{-1}(i), 0} = \emptyset$. This implies the following simple claim (proof omitted).

\begin{claim}
\label{claim}
Let $f, k, Q, G, \pi$ be as in Definition~\ref{def:per_prot}. If $f$ is $\pi$-symmetric, then $\pi(Q)$ is a correct $\NOF_{G_\pi}$ protocol for $f$.
\end{claim}

\begin{definition}
\label{def:filtering}
Let $G = (V, E)$ be a restriction graph for $k$ parties. A set of triplets $$S \subseteq \{(a, b, B) |  a \neq b \in [k], B \subseteq \overline{N}(a), a,b \notin B \}$$ is a filtering set for $G$, if for any two distinct $(a, b, B), (a', b', B') \in S$, we have that $a \notin (\{b'\} \bigcup B')$ and $(\{b\} \bigcup B) \bigcap (\{b'\} \bigcup B') = \emptyset$. For a filtering set $S$, denote $R(S) = \max_{a \in [k]} \sum_{(a, b', B') \in S} |B'|$. A filtering set for $G$ is an $\ell$-filtering set, if $R(S) \leq \ell - 1$.
\end{definition}

Let $f$ be a $k$-argument function and $G = (V, E)$ be a restriction graph for $k$ parties. We now have two, somewhat similar, definitions: an $\ell$-filtering set for $G$ and a $\MYPRODG$-multiplexing set. We discuss the difference between them.
Finding a $\MYPRODG$-multiplexing set may be difficult, as the definition requires us to have a sequence of $\ell$ permutations. Furthermore, we need to have $\ell$ oblivious $\NOF_{G_{\pi_r}}$ protocols which computes $f$, a protocol for each permutation $\pi_r$, such that they all have the same $\LEN$ (up to the permutation).
The multiplexing set, is a set of potential triplets, representing which protocols and which messages in those protocols, we can multiplex.
Instead, for the special case of symmetric functions, to get a multiplexed protocol, we only need to have one oblivious $\NOF_{G}$ protocol. The filtering set will have triplets, representing multiplexing candidates, based only on the graph $G$, and the permutations themselves will be computed from the filtering set.

\begin{theorem}
\label{thm_first}
Let $f$ be a $k$-argument symmetric function, $G = (V, E)$ a restriction graph for $k$ parties, $Q$ an oblivious $\NOF_G$ protocol for $f$ and $S$ an $\ell$-filtering set for $G$. Then, $$D^{\NOF}(f^{\ell}) \leq \ell \cdot D^{\NOF_G}(Q) + \ell - \sum_{(a, b, B) \in S} |B| \cdot W_Q(a,b).$$
\end{theorem}

\label{app-proof}

\begin{proof}

We start by presenting an overview of the proof. Our aim is to form a sequence of $\ell$ permutations $\pi_{r}$ over $[k]$, written as an $\ell \times k$ matrix, where each row is a permutation. We then take those permutations and find for each permutation $\pi_r$, a $\NOF_{G_{\pi_r}}$ protocol $Q_{\pi_r}$ which is $\pi_r$-pattern-robust with respect to $Q$, and which solves $f$. We then convert the given $\ell$-filtering set $S$ into a respective $\MYPRODG$-multiplexing set $S^{*}$, for the aforementioned permutations and apply Theorem~\ref{thm:second} to conclude the proof.

We construct an $\ell \times k$ matrix, $A_{G, \ell, S}$ (or $A$ in short), such that the first row of $A$ is the identity permutation over $[k]$. Let $S = \{(a_1, b_1, B_1), \ldots, (a_{|S|}, b_{|S|}, B_{|S|})\}$ be the given $\ell$-filtering set.
Each triplet $(a_i, b_i, B_i)$, will partially fill entries in $A$, so that in the corresponding permutations, $b_i$ is mapped to elements of $B_i$. Specifically, denote $last_i = 1 + \sum_{u < i | a_i = a_u}|B_u|$, that is, the last index of a row in $A$, which we filled with entries having to do with $a_i$. Assume $a_i = a_u$ and, without loss of generality, $i < u$, then for $j=1, \ldots, |B_i|$ its follows from the definition of $last$ that $last(i) + j \leq last(u)$. For $j = 1, \ldots, |B_i|$, denote  $row(i, j) = last_i + j$. For the triplet $(a_i, b_i, B_i)$, for $1 \leq j \leq |B_i|$, entries will be added to rows $row(i, j)$. Eventually, $row(i, j)$, will hold a permutation, where $a_i$ will be mapped to itself and $b_i$ will be mapped to $(B_i)_j$. For example, if all $a_{i}'s$ are different, then $last_i = 1$ for all $i$. If, for example, all the $a_{i}'s$ are the same, each $row(i, j)$ is a different row in the matrix. See Example~\ref{example_multilexing_set} below for a set $S$, its corresponding matrix and the corresponding function $row$.

More formally, for $(a_i, b_i, B_i) \in S$, for $j = 1, \ldots, |B_i|$ we have $A[row(i, j)][a_i] = a_i$ and $A[row(i, j)][b_i] = (B_{i})_{j}$. Note that since $R(S) \leq \ell - 1$, we have that $row(i, j) \leq l$, and so, $\ell$ rows suffice for $A$.
For a columns set $T \subseteq [k]$, we will use the notation $A[i][T]$ for $\{A[i][j] : j \in T\}$.
We claim that the above $A$ can be extended to accommodate the following restrictions:
\begin{enumerate}
	\item Each row of $A$ is a permutation over $[k]$.
	\item For each $1 \leq i \leq |S|, 1 \leq j \leq |B_i| : A[row(i, j)][B_i] = B_i \union \{b_i\} \setminus \{(B_{i})_ j\}$.
\end{enumerate}
First we argue that so far, there is no row in $A$ with a value appearing more than once. Since the sets $\left\{b_i\right\} \bigcup B_i$ are always disjoint (from one another, and from $\{a_i\}$'s, as $S$ is a filtering set), the only possible repetition can be between $a_i's$. That is, there are $i_1 < i_2, j_1, j_2$ such that $1 \leq j_1 \leq |B_{i_1}|, 1 \leq j_2 \leq |B_{i_2}|, a_{i_1} = a_{i_2}$ and $row(i_1, j_1) = row (i_2, j_2)$. That is, $last(i_1) + j_1 = last(i_2) + j_2$. As explained before, $last(i_1) + j_1 \leq last(i_2) < last(i_2)+ j_2$, so that leads which is a contradiction.
Given that currently $A$ does not have any repeating value in any of its rows, and that the above two requirements of the extension, only pose restrictions per row, on values which currently do not appear in the row, it is obvious that $A$ can be completed as required.

From now on, we can assume that for each $1 \leq r \leq \ell$, the row $A[r]$ is a permutation which will be denoted as $\pi_r$ and the protocol $Q_{\pi_r}$ is a suitable $\NOF_{G_{\pi_r}}$ protocol for computing $f$, as guaranteed by Claim~\ref{claim}. That is, $Q$ itself is $\pi_r$-pattern-robust.

For $(a_i, b_i, B_i) \in S$, let $R_i = \{row(i, j) : 1 \leq j \leq |B_i| \}$ and $S^* = \{ (a_i, b_i, R_i) : (a_i, b_i, B_i) \in S \} $. We prove that $S^*$ is a $\MYPRODG$-multiplexing set.  We start by proving that for every $1 \leq i \leq |S|$, the triplet $(a_{i}, b_{i}, R_{i})$ is $\MYPRODG$-good.

Condition~(1) of Definition~\ref{def:good} is met, as $A$ was built such that $A[row(i, j)][a_i] = a_i$.
$\MAP_{\MYPROD, b_{i}, R_{i}}$ is simply $\{\pi_{r}(b_{i}) : r \in R_{i} \} =  \{\pi_{r}(b_{i}) : r \in \{row(i, j) : 1 \leq j \leq |B_{i}| \} \} =  \{\pi_{row(i, j)}(b_{i}) : 1 \leq j \leq |B_{i}| \} $. Due to the construction of $row$, this is $B_{i}$. It follows that  $|\MAP_{\MYPROD, b_i, R_i}|$ is in fact $|B_i|$, which by the definition of $R_i$, and the way $A$ is built, is also $|R_i|$.
As the given $S$ is a filtering set, it follows that $B_{i} \subseteq \overline{N_{G}(a_i)}$. It also follows that $b_i \notin B_i$, that is $b_i \notin \MAP_{\MYPROD, b_{i}, R_{i}}$. Thus, condition~(2) of Definition~\ref{def:good} is met.
It is left to prove condition~(3) in Definition~\ref{def:good}. Let $r \in R_i$, we want to prove that $\MAP_{\MYPROD, b_i, R_i} \bigcup \{b_i\} \backslash \{\pi_r(b_i)\} \subseteq \overline{N_{G_{\pi_r}}(a_i)}$. That is, $B_{i} \bigcup \{b_i\} \backslash \{\pi_r(b_i)\} \subseteq \overline{N_{G_{\pi_r}}(a_i)}$. Which is equivalent to $B_{i} \bigcup \{b_i\} \backslash \{\pi_r(b_i)\} \subseteq  \overline{\pi_r(N_{G}(a_i))}$ (We use the fact that $\pi_r(a_i) = a_i$, as this is how $A$ was built). The way $A$ is built makes sure that  $\pi_r(B_i) = B_{i} \bigcup \{b_i\} \backslash \{\pi_r(b_i)\}$. So we need to show that $\pi_r(B_i) \subseteq {\pi_r(N_{G}(a_i))}$. This simply follows from $B_{i} \subseteq \overline{N_{G}(a_i)}$).

Let $(a_{i_1}, b_{i_1}, R_{i_1})$, $(a_{i_2}, b_{i_2}, R_{i_2})$ be two distinct triplets in $S^*$ and let $(a_{i_1}, b_{i_1}, B_{i_1})$, $(a_{i_2}, b_{i_2}, B_{i_2})$ be the respective triplets in $S$. That is, $a_{i_1} \neq b_{i_1}$ and $a_{i_2} \neq b_{i_2}$.  To prove that $S^{*}$ is a $\MYPRODG$-multiplexing set, we need to prove that $a_{i_1} \notin \{b_{i_2}\} \bigcup \MAP_{\MYPROD, b_{i_2}, R_{i_2}}$ and $(\{b_{i_1}\} \bigcup \MAP_{\MYPROD, b_{i_1}, B_{i_1}}) \bigcap (\{b_{i_2}\} \bigcup \MAP_{\MYPROD, b_{i_2}, R_{i_2}}) = \emptyset$. The conditions translate to $a_{i_1} \notin \{b_{i_2}\} \bigcup B_{i_2}$ and $\{b_{i_1}\} \bigcup B_{i_1}) \bigcap (\{b_{i_2}\} \bigcup B_{i_2}) = \emptyset$, which is true as $S$ is a filtering set.

We can now apply Theorem~\ref{thm:second}, and get the desired result.
\qed
\end{proof}

\bigskip

The following example demonstrates the construction of a matrix $A$ as in the above proof.

\begin{example}
\label{example_multilexing_set}

Let $G = (V, E)$ be a restriction graph for $9$ parties, where the set of edges $E = \left\{(1, 2), (1, 5), (7, 8)\right\}$
\end{example}

Take the ordered set $S_1 = (1, 2, \left\{3, 4\right\}),  S_2 = (1, 5, \left\{6\right\}), S_3 = (7, 8, \left\{9\right\})$.
As $R(1) = 3, R(7) = 1$ and both are bound by $\ell - 1 = 3$. So the set is a $4$-filtering set of size $3$.
The appropriate $row's$ are $row(1, 1) = 2, row(1, 2) = 3, row(2, 1) = 4, row(3, 1) = 2$. The corresponding matrix is
$$A =
\left( {\begin{array}{*{20}c}
   1 & 2 & 3 & 4 & 5 & 6 & 7 & 8 & 9 \\
   1 & 3 & ? & ? & ? & ? & ? & 9 & ? \\
   1 & 4 & ? & ? & ? & ? & ? & ? & ? \\
   1 & ? & ? & ? & 6 & ? & ? & ? & ? \\
 \end{array} } \right)
$$

\begin{example}
We show that Theorem~\ref{thm_first} implies Proposition~\ref{prop_eq_multiple}. Let $k$ be as in Proposition~\ref{prop_eq_multiple}. Let $G = ([k], [k] \times [k] \setminus \left\{(i, i) : i \in [k]\right\} \setminus \left\{(k, 4), \ldots, (k, k - 1) \right\})$ be a restriction graph. I.e, $G$ is a complete digraph, except that it doesn't have edges from the node $k$ to any of the nodes in $[4, k - 1]$. We describe a $\NOF_G$ protocol $Q$ for $f$: party $P_{k}$ sends $P_1$ the bit $1$ if $x_1 = x_2$, and $0$ otherwise. If $P_2$ got $1$ and all other inputs are equal, it outputs $1$, otherwise it outputs $0$. For the $\frac{k - 3}{2}$-size $(\frac{k - 1}{2})$-filtering set $S = \left\{(k, 2, \left\{4, 6, \ldots k - 1\right\})\right\}$, we apply Theorem~\ref{thm_first}. It follows that $D^{\NOF}(EQ_{k}^{\frac{k - 1}{2}}) \leq \frac{k - 1}{2} \cdot 1 + \frac{k - 1}{2} - 1 \cdot(\frac{k - 3}{2}) = 1 + \frac{k - 1}{2}$.
\end{example}

\medskip

\section{Multiplexing NOF protocols via the Myopic Multi-Party Model}
\label{app-myopic}

In this section, we show that given so-called $\MYOPIC$ protocols for $\ell$ permutations of their $k$ parties, that satisfy certain combinatorial requirements, we can ``multiplex'' those protocols into one $\NOF$ protocol, obtaining some savings in communication complexity. We start with some definitions.

\begin{definition}
The $k$-party one-way $\NOF$ model is identical to the $\NOF$ model, except that the communication pattern is restricted to the following ``one-way'' manner: the parties are organized according to some permutation $\pi$ and each party, in its turn, sends a single message to the next party according to $\pi$ (namely, party $P_{\pi(i)}$ sends a message to $P_{\pi(i+1)}$).  The last party announces the output. In the {\em myopic} one-way model, there is an additional restriction
that party $P_{\pi(i)}$ sees only the inputs of the pervious parties and the following party. That is, $P_{\pi(1)}, \ldots, x_{\pi(i-1)}$ and $x_{\pi(i+1)}$.  For a permutation $\pi$ over $[k]$, we denote $\MYOPIC_{\pi}$ as the myopic model with $\pi$ being the order of parties.
\end{definition}

As an example of a myopic protocol, see the protocol in the proof of Proposition~\ref{prop_eq}, where instead of $P_{k-1}$ broadcasting its message, it sends its message to $P_k$ and $P_k$ announces the output.

The myopic model was first discussed in~\cite{G06}, where it was investigated with regards to the Pointer Jumping problem. The additional restriction, compared to the one-way model, was used to show better lower bounds on the communication complexity of randomized protocols. It was further studied in~\cite{C07}, where better lower bounds were shown.

 A $\MYOPIC_{\pi}$ protocol $Q$, can be thought of as a special case of a restricted-$\NOF$ protocol, for the graph $G_{\pi}$, where $G = ([k], \{(i, j) : (j < i) \vee (j = i + 1)\})$. In addition, another restriction on $Q$, is that messages are sent in a ``one-way'' fashion, according to $\pi$. That is, if $|m^{Q}_{i \rightarrow j, t}| > 0$ for some input $\vec{x}$, then $\pi(t) = i, \pi(t+1) = j$. Therefore, the next result may be thought of as corollary of Theorem$~\ref{thm_first}$, except that, we do not require the protocol $Q$ to be oblivious (nor do we require that other protocols which will be used, will have the same communication pattern as $Q$). Thus, we get a more general result.

We now describe, informally, which restrictions are required of two $\MYOPIC$ protocols $Q_{\pi_1}, Q_{\pi_2}$ to be able to multiplex messages sent in them. Assume we have $k$ parties and we want to multiplex messages sent by $P_i$. So, $P_i$ needs to appear as the $t^{th}$ party for some $t \in [k - 1]$, in $Q_{\pi_1}, Q_{\pi_2}$. We also need the parties which receive the messages sent by $P_i$ in those permutations to be different. That is, we need $|\{\pi_1(t + 1), \pi_2(t + 1)\}| = 2$. Next, in analogy to the requirements for multiplexing in the restricted model, the message $P_i$ sends $P_{\pi_1(t + 1)}$ in $Q_{\pi_1}$, should not depend on $x_{\pi_2(t+1)}$ (similarly for $P_{\pi_{2}(t + 1)}$ in $Q_{\pi_2}$ and $x_{\pi_1(t+1)}$). To accomplish this, we require that $P_{\pi_2}(t + 1)$ does not appear in $\pi_1$ before the $t^{th}$ place.
Lastly, we need that messages sent to $P_i$ in $Q_{\pi_1}$ and $Q_{\pi_2}$, will not be multiplexed, so in the combined $\NOF$ protocol, these messages will appear plainly on the board, so $P_{\pi_{1}(t + 1)}, P_{\pi_{2}(t + 1)}$ can see them. These requirements can be extended to accommodate multiplexing $\ell$ protocols.

Since we don't assume that $Q_{\pi_1}, Q_{\pi_2}$ were oblivious, the savings we have made in the communication complexity, are per input $\vec{x}$, and therefore the resulting Theorem has a different expression than the one in Theorem~\ref{thm:second}, for the bound on communication complexity of the direct-sum.

We formulate the necessary requirements for $\MYOPIC$ protocols to be multiplexed, in an analogy to Definition$~\ref{def:multiplexing}$.
\medskip

\begin{definition}
Let $\MYPROD = \pi_1, \ldots, \pi_\ell$ be a sequence of $\ell$ permutations over $[k]$. A triplet $(i, s, U)$ is $\MYPROD$-binding if $i, s \in [k]$ and $U \subseteq [\ell]$ such that:
\begin{enumerate}
\item For all $u \in U : \pi_{u}(i) = s$.
\item $|\pi_{u}(i+1) : u \in U| = |U|$.
\item $\{ \pi_{u}(i') : u \in U, i'<i \} \bigcap \{ \pi_{u}(i+1) : u \in U \} = \emptyset$.
\end{enumerate}

The definition guarantees a sequence of $|U|$ permutations, such that for each of the permutations, in step $t$, the party $P_s$ is the one who sends a message. Also guaranteed, is that all the recipients of those messages, according to the $|U|$ permutations, are parties which have not appeared before $P_s$ in any of the $|U|$ permutations.
\end{definition}

\begin{definition}
Let $\MYPROD$ be a sequence of $\ell$ permutations over $[k]$. A $\MYPROD$-repetitive set $S$ is a set of $\MYPROD$-binding triplets, such that for all distinct $(t_1, s_1, U_1), (t_2, s_2, U_2) \in S:$

\begin{enumerate}
\item $s_1 \notin \{ \pi_{u}(t_2 + 1) : u \in U_2 \}$.
\item if $(t_1, s_1) = (t_2, s_2)$ then $U_1 \bigcap U_2 = \emptyset$.
\end{enumerate}
\end{definition}

\begin{theorem}
\label{thm:third}
Let $\MYPROD$ be a sequence of $\ell$ permutations over $[k]$. Let $f$ be a $k$-argument function and $S$ a $\MYPROD$-repetitive set. Let $Q_{\pi_1}, \dots, Q_{\pi_\ell}$ be $\ell$ protocols corresponding to the computation of $f$ in $\MYOPIC_{\pi_i}$. Also assume that messages sent during those protocols are prefix-free. Denote $Cost_{\vec{x}, u, t}$ as the number of bits sent in the message from $\pi_u(t)$ to $\pi_u(t+1)$, in protocol $Q_{\pi_u}$, on input $\vec{x}$.
For a $\MYPROD$-binding triplet $(t, s, U)$ denote $TCost_{\vec{x}, U, t} = \sum_{u \in U} Cost_{\vec{x}, u, t}$, meaning, the total number of bits sent from party $P_s$ to parties in $\{\pi(u)(t + 1) : u \in U\}$, in the corresponding protocols. Also denote $MCost_{\vec{x}, U, t} = \max_{u \in U} Cost_{\vec{x}, u, t}$, as the maximal length of a message, from the aforementioned messages. Then,
$$D^{\NOF}(f^{\ell}) \leq \max_{\vec{x}}(\sum_{u = 1 \ldots l}\sum_{t = 1 \ldots \ell - 1} Cost_{\vec{x}, u, t} - \sum_{(a, s, U) \in S} (TCost_{\vec{x}, U, a} - MCost_{\vec{x}, U, a})).$$
\end{theorem}

\begin{proof}
We describe a protocol $Q'$ for $f^\ell$ in the $\NOF$ model. We will use protocol $Q_i$ to solve $f$ on the $i^{th}$ instance of $f^\ell$. Protocol $Q'$ will run all those protocols in parallel and will perform multiplexing when possible. In round $t$, for each instance $i$, the corresponding step of protocol $Q_i$ is made on the $i^{th}$ instance, only that instead of actually sending the message on a private channel, as in the myopic model, the corresponding party writes it on the board. To reduce the amount of communication, when performing round $t$, for every $(t, s, U) \in S$, instead of $P_s$ sending $|U|$ messages to each $\pi_u(t + 1)$ for $u \in U$, party $P_s$ xors those $|U|$ messages together, and writes on the board the xored message. We argue that each of the original recipients can then restore the original message it was supposed to receive.
We take $u_1 \in U$, and show that party $\pi_{u_1}(t + 1)$ can retrieve its message. We need to show that for each $u \in \{U \backslash u_1\}$, party $\pi_{u_1}(t + 1)$ can compute the message sent at round $t$ from $P_s$ to $P_{\pi_u(t+1)}$ in $Q_u$. This message may depend, by the definition of $\MYOPIC_{\pi_u}$, on inputs of $\{P_{\pi_u(t')}| t' \leq t + 1, t' \neq t \}$ and on the message received by $P_s$ in $Q_u$. By condition~(3) of the definition of a $\MYPROD$-repetitive set, this is disjoint to $\{P_{\pi_{u_1(t + 1)}}\}$ (as $u, u_1 \in U$). That is, they are known to $P_{\pi_{u_1(t + 1)}}$ in the $\NOF$ model. The received message of $P_s$ in $Q_u$ is not multiplexed (as for all $(t, s, U) \neq (t_2, s_2, U_2) \in S$ it is promised that $s \notin \{\pi_u(t_2 + 1) : u \in U_2\}$), and therefore written on the board. It follows that these messages can be computed by $u_1$ in the $\NOF$ model, and that it can also compute its original message, by xoring all of them with the xored valued on the board. Note that for the case when its original message was shorter then the other multiplexed messages, it still need to know how many bits of the xored result it needs to take. For that we use the fact that messages in the protocol are prefix-free.

That is, for every $(t, s, U) \in S$, we manage to send a single message only, instead of $|U|$, when the length of the message sent is the longest massage out of the possible $|U|$ (the shorter ones can be padded to fit the length of the longest one).
\qed
\end{proof}

\begin{corollary}
\label{corol:two}
Let $\MYPROD$ be a sequence of $\ell$ permutations over $[k]$. Let $f$ be a $k$-argument function and $S$ a $\MYPROD$-repetitive set. Let $Q_{\pi_1}, \dots, Q_{\pi_\ell}$ be $\ell$ oblivious protocols corresponding to the computation of $f$ in $\MYOPIC_{\pi_i}$, such that they all have the same communication pattern. Then, $$D^{\NOF}(f^{\ell}) \leq \ell \cdot D^{\MYOPIC_{\pi_1}}(Q_1) + \ell - \sum_{(a, b, U) \in S} (|U| - 1) \cdot W_{Q_1}(a,b).$$

\end{corollary}

Note that Corollary~\ref{corol:two} can also be viewed as a a Corollary to Theorem~\ref{thm:second}, as for the case of oblivious protocols, Theorem~\ref{thm:third} is a private case of Theorem~\ref{thm:second}.

\end{document}